\documentclass[letterpaper, 10 pt, conference]{ieeeconf}
\usepackage{cite}
\usepackage{amsmath,amssymb,amsfonts}
\usepackage{algorithmic}
\usepackage[ruled]{algorithm2e}
\usepackage{graphicx}
\usepackage{textcomp}
\usepackage{amsmath}
\usepackage{dsfont}
\usepackage{bbold}
\usepackage{tikz}
\usepackage{xcolor}
\usepackage{verbatim}
\usetikzlibrary{arrows}

\newtheorem{definition}{Definition}

\newtheorem{remark}{Remark}
\newtheorem{problem}{Problem}

\newtheorem{lemma}{Lemma}
\newtheorem{proposition}{Proposition}

\newtheorem{theorem}{Theorem}

\DeclareMathOperator*{\argmin}{arg\,min}
\newcommand{\mc}{\mathcal}

\newcommand{\be}{\begin{equation}}
\newcommand{\ee}{\end{equation}}
\newcommand{\ba}{\begin{array}}
\newcommand{\ea}{\end{array}}

\newcommand{\ds}{\displaystyle}

\newcommand{\1}{\mathbf 1}
\newcommand{\I}{\mathbf I}
\newcommand{\Prob}{\mathbb P}
\newcommand{\E}{\mathbb E}

\newcommand{\supp}{\text{supp}}

\IEEEoverridecommandlockouts                             
\title{Information design in Bayesian routing games}

\author{Leonardo Cianfanelli, Alexia Ambrogio and Giacomo Como
\thanks{The authors are with Dipartimento di Scienze Matematiche, Politecnico di Torino. Giacomo Como is also with  the Department of Automatic Control, Lund University (e-mail: leonardo.cianfanelli@polito.it, alexia.ambrogio@studenti.polito.it, giacomo.como@polito.it, webpage: https://sites.google.com/view/leonardo-cianfanelli/, https://staff.polito.it/giacomo.como/)
}}

\begin{document}

\maketitle
\thispagestyle{empty}
\pagestyle{empty}

\begin{abstract}
	We study optimal information provision in transportation networks when users are strategic and the network state is uncertain. An omniscient planner observes the network state and discloses information to the users with the goal of minimizing the expected travel time at the user equilibrium. Public signal policies, including full-information disclosure, are known to be inefficient in achieving optimality. For this reason, we focus on private signals and restrict without loss of generality the analysis to signals that coincide with path recommendations that satisfy \emph{obedience constraints}, namely users have no incentive in deviating from the received recommendation according to their posterior belief. We first formulate the general problem and analyze its properties for arbitrary network topologies and delay functions. Then, we consider the case of two parallel links with affine delay functions, and provide sufficient conditions under which optimality can be achieved by information design. Interestingly, we observe that the system benefits from uncertainty, namely it is easier for the planner to achieve optimality when the variance of the uncertain parameters is large. We then provide an example where optimality can be achieved even if the sufficient conditions for optimality are not met.
\end{abstract}

\begin{keywords}
	Information design, Bayesian routing games.
\end{keywords}

\section{Introduction}\label{sec:introduction}

Routing games describe the behaviour of a non-atomic set of strategic users travelling in transportation networks \cite{wardrop1952road}. Each link of the network is endowed with a delay function that returns the travel time along the link as a function of the flow, and users aim at selecting paths with minimum travel time. A user equilibrium of the game is a flow such that no user has an incentive in deviating to other paths from the one she has selected. Typically, user equilibria are suboptimal in terms of total travel time on the network. The goal of a network planner is then to influence the user behaviour by incentive mechanisms to steer the equilibrium of the game towards the system optimum. Standard approaches are to charge tolls \cite{cole2006much,maggistro2018stability}, or to intervene on the network structure \cite{cianfanelli2023optimal,farahani2013review}. When the state of the network is uncertain (we refer to this class of games by \emph{Bayesian routing games}), this uncertainty can be leveraged by an omniscient planner that observes the state of the network and discloses information to the users. How to disclose this information to minimize the system cost at equilibrium is called in the literature \emph{information design problem}.

The role of information in routing games has been widely studied in the literature. Informational Braess' paradox shows that the travel time at the equilibrium may decrease when a fraction of users are not aware of some roads of the network \cite{acemoglu2018informational}. In \cite{wu2017informational,wu2021value} the authors analyze how the efficiency of the equilibrium varies in terms of the fraction of adopters of information providers, showing that too much information might hurt the system. The information design problem in routing games has been first considered in \cite{daskam17}, where the authors show by examples that in routing games revealing full-information and all public signal policies are in general suboptimal and providing private information is preferable. In \cite{tavafoghi2017informational} it is studied an information design problem with two parallel links and affine delay functions where the free-flow delay of one link is a discrete random variable, and sufficient conditions for optimality of private information are provided.
In \cite{savla}, the authors focus on computational perspectives and propose an algorithm to solve efficiently the information design problem with polynomial delay functions and arbitrary network topologies. Dynamical information provision is studied in \cite{dughmi2016algorithmic,meigs2020optimal,tavafoghi2019strategic}.

Motivated by the suboptimality of public signal policies, in this paper we consider private information provision in Bayesian routing games. Compared to \cite{tavafoghi2017informational}, we formulate Bayesian routing games and the corresponding information design problem for arbitrary network topologies and arbitrary delay functions described by continuous random variables. We show how to equivalently reformulate the problem (due to the revelation principle \cite{bergemann16a,bergemann19}) by restricting the signal space to path recommendations such that users have no incentive in deviating from the received recommendation, which is known as \emph{obedience constraint}. The policy then specifies the fraction of users that receive each recommendation as a function of the network state. The users are not aware of the recommendations received by other ones, but the prior distribution of the network state and the information policy are publicly known.

After analyzing the properties of the general problem, we focus on the special case of two links with affine delay functions, which is a convex problem. We show that under a restrictive assumption on the support of the random variables that describe the delay functions, there exist sufficient and necessary conditions on the moments of the random variables under which the optimality can be achieved by information design. Interestingly, we show that the uncertainty on the network state is beneficial for the system under the optimal information policy, namely the planner is able to persuade the users to distribute according the system optimum flow, which typically does not happen when users have full-information on the network state. We also remark that the approach can be generalized to arbitrary topologies although the corresponding information design problem is in general non-convex, which is left for future research. Finally, we analyze a particular case that allows explicit computation, showing that optimality can be achieved even if the restrictive assumption on the support of the random variables is not met.

The paper is organized as follows. In Section \ref{sec:model,problem} we present Bayesian routing games, formulate the information design problem, and analyze its properties. We then present the results in Section \ref{sec:results}. In Section \ref{sec:conclusion} we summarize the contribution and discuss future research lines.

\subsection{Notation}
Given a vector $x$, we let $x'$ denote its transpose. $\1$, $\I$ and $\delta^{(i)}$ denote the vector of all ones, the identity matrix the vector with $1$ in $i$-th entry $\delta_i^{(i)}$ and $0$ in all the other positions, whose size may be deduced from the context.
For a finite set $\mc P$, let $\Delta_{\mc P}$ denote the simplex on $\mc P$, i.e.,
$$
\Delta_{\mc P} = \{x \in \mathds R_+^{\mc P}: \1'x = 1\}.
$$
Given $x,a,b$ in $\mathds R$, with $a\le b$, we let
$$
[x]_a^b=\max\{a,\min\{x,b\}\} = \begin{cases}
	a \ &\text{if} \ x < a \\
	x \ &\text{if} \ x \in [a,b] \\
	b \ &\text{if} \ x > b. \\
\end{cases}
$$
\section{Model and problem formulation}\label{sec:model,problem}
We present the model in Section \ref{sec:model}, and formulate the problem in Section \ref{sec:problem}.
\subsection{Model}
\label{sec:model}
We model the transportation network topology as a directed multigraph $\mc{G}=(\mc{N},\mc{E})$, where $\mc{N}$ is the set of nodes and
$\mc{E}$ is the set of links. Every link $e$ in $\mc E$ is directed from its tail node $\sigma(e)$ in $\mc N$ to its head node $\xi(e)$ in $\mc N$. The transportation network topology is fully characterized by its node-link incidence matrix $B$ in $\mathds{R}^{\mc N \times \mc E}$, with entries $B_{ne} = 1$ if $\sigma(e)=n$, $-1$ if $\xi(e)=n$, or $0$ otherwise. 

We denote by $o$ and $d$ in $\mc{N}$, with $o\ne d$, the origin and, respectively, the destination nodes. We assume that $d$ is reachable from $o$. 
Let a unitary mass of non-atomic users travel from $o$ to $d$, and define $\nu = \delta^{(o)}-\delta^{(d)}$. We then let 
$$
\mc F = \{f \in \mathds{R}^{\mc E}_+: Bf=\nu\}
$$
be the set of unitary $o$-$d$ flows, where a flow describes how users are distributed across links.

The state of the network is described by a random variable $\theta$ taking values on a probability space $(\Theta, \mathcal{A}, \Prob)$.
Each link of the network is endowed with a delay function $\tau_e : \mathds R_+ \times \Theta \to \mathds R_+$ that depends on the flow and on the realization of the network state. We shall assume that $\tau_e(f_e,\theta)$ is increasing in $f_e$ for every $\theta$ in $\Theta$ to consider congestion effects.

A network flow is a function $f : \Theta \to \mc F$ that assigns to each network state $\theta$ in $\Theta$ a vector $f(\theta)$ in $\mc F$. 
We define the system cost as the expected travel time on the network, i.e.,
$$
C(f) = \int_\Theta \sum_{e \in \mc E} f_e(\theta) \tau_e (f_e(\theta),\theta) d\Prob(\theta)\,. 
$$
Network flows depend on the available information on the network state. A setting of interest is the \emph{full-information} setting, where all users know exactly the realization of the network state $\theta$.
Given the information setting, we can define two solution concepts, the \emph{system optimum} flow and the \emph{user equilibrium} flow. System optimum flows describe a scenario where each user is forced to select a path with the goal of minimizing the system cost. However, users are typically strategic and aim at minimizing their own travel time. This behavior is captured by the notion of user equilibrium flow. 

In full-information setting, flows are allowed to depend on the network state $\theta$, which is known to the users. We can imagine this as the result of an omniscient planner observing the network state and revealing this information to the users. In this setting, the system optimum is a network flow that depends on $\theta$ and minimizes the system cost:
\be\label{eq:system_opt_full}
	f^* \in \argmin_{f: \Theta \to \mc F} \int_\Theta \sum_{e \in \mc E} f_e\left(\theta\right) \tau_e (f_e(\theta),\theta) d\Prob(\theta).
\ee
Observe that in full-information setting, $f^*$ is independent of the prior distribution of $\theta$, since $f^*(\theta)$ is the unitary $o$-$d$ flow that minimizes $\sum_{e \in \mc E} f_e\left(\theta\right) \tau_e (f_e(\theta),\theta)$ for a given network state $\theta$.
To formalize the notion of user equilibrium, let $\mc P$ denote the set of paths from $o$ to $d$ and let $z: \Theta \to \Delta_{\mc P}$ denote a path flow, describing how users distribute on the paths given the network state. Let $A$ in $\mathds R^{\mc E \times \mc P}$ be the link-path incidence matrix, with entries $A_{ei} = 1$ if link $e$ belongs to path $i$, or $0$ otherwise. A path flow $z$ induces a unique network flow $f(\theta) = Az(\theta)$, and the cost of a path $i$ for a network state $\theta$ is the sum of the delay functions of the links along the path, i.e.,
$$
c_i (f(\theta),\theta) = \sum_{e \in \mc E} A_{ei} \tau_e(f_e(\theta),\theta).
$$ 
The full-information user equilibrium is a network flow $f^W: \Theta \to \mc F$ that admits $z^W: \Theta \to \Delta_{\mc P}$ such that $f^W(\theta) = Az^W(\theta)$ and such that, for every $i$ in $\mc P$ and $\theta$ in $\Theta$, $z_i^W(\theta)>0$ implies
\begin{equation*}
	c_i (f^{W}(\theta),\theta) \le
	c_j (f^{W}(\theta),\theta) \quad \forall j \in \mc P,
\end{equation*}
In other words, the full-information user equilibrium is a network flow such that, for each network state $\theta$, each used path has optimal cost, which means that no user has an incentive in deviating to other paths.

In general, it is known that revealing full-information is suboptimal for the system, namely the full-information user equilibrium is suboptimal compared to the full-information system-optimum, which is by definition the best network flow in terms of system cost. Moreover, it is known that any public signal is inefficient in achieving optimality \cite{daskam17,tavafoghi2017informational,savla}. For this reason, in the next part of the paper we shall consider how to design optimal private signal policies.

\subsection{Problem formulation}\label{sec:problem}
We assume that an omniscient planner observes the realization of the network state, and given this observation sends private signals to the users, with the goal of minimizing the expected total travel time of the user equilibrium flow. A signal policy is a map $\pi : \Theta \to \Delta_{\mc P}$ that assigns to each network state the fraction of users that is recommended to take each path. One might consider more general types of signals, however assuming that signals coincide with path recommendations is without loss of generality, as stated by the revelation principle \cite{bergemann19,bergemann16a}. 

A fundamental assumption is that the signal policy and the prior distribution of the network state are known to the users. Once the signals have been sent, users update their belief on the network state. Let $d\Prob_i$ denote the posterior belief of users that receive recommendation $i$, which is computed according to Bayes' formula, i.e.,
\begin{equation}\label{eq:posterior}
	d\Prob_i(\theta) = \frac{\pi_i (\theta) d\Prob(\theta)}{\int_{\Theta} \pi_i (\omega) d\Prob(\omega)}\,,
\end{equation}  
and let $\E_i[\cdot]=\int_{\Theta}\cdot\, d\Prob_i(\theta)$ denote the expected value after receiving signal $i$, i.e., according to posterior belief \eqref{eq:posterior}. 

As the users update their beliefs on the network state, the user equilibrium varies accordingly. Indeed, users follow the recommendation only if it is a best response to do so. Let $y$ in $\mathds R_+^{\mc P \times \mc P}$ be a matrix such that $y\1 = \1$, whose element $y_{ij}$ denotes the fraction of users that choose path $j$ among those who receive recommendation of taking path $i$. We can interpret $y$ as the action distribution of a population game with multiple populations that differ in the received signal. Given $\pi$ and $y$, the network flow depends on the network state $\theta$ via
$$
f_e^{\pi,y}(\theta) = \sum_{j \in \mc P} A_{ej} \sum_{i \in \mc P} \pi_i(\theta) y_{ij}\,,\qquad\forall e\in\mc E\,,
$$
or, more compactly, 
$$f^{\pi,y}(\theta)=Ay'\pi(\theta)\,.$$
\begin{definition}[Bayesian user equilibrium]
	Given a policy $\pi$, a flow $f^{\pi,y}(\theta)=Ay'\pi(\theta)$ is a Bayesian user equilibrium if $\pi_i(\theta)y_{ij}>0$ for at least a $\theta$ in $\Theta$ and $i,j$ in $\mc P$ implies that 
$$
\E_i[c_j(f^{\pi,y}(\theta),\theta)] \leq \E_i[c_k(f^{\pi,y}(\theta),\theta)] \quad \forall k \in \mc P.
$$
\end{definition}\smallskip

In other words, we require that if some users receive recommendation $i$ and travels on path $j$, then path $j$ has to be optimal according to the posterior belief of users that have received recommendation $i$. The next result characterizes Bayesian user equilibria.

\begin{proposition}\label{prp:pot}
	Given a policy $\pi$, a flow $f^{\pi,y^*(\pi)}(\theta)=A(y^*(\pi))'\pi(\theta)$
	is a Bayesian user equilibrium if and only if $y^*(\pi)$ is a solution of the convex program
	\begin{equation}\label{eq:y} 
		y^*(\pi) \in \argmin_{y \in \mathds R^{\mc P \times \mc P}_+: y\1 = \1} \Phi_\pi \left(y\right),
	\end{equation}
	where
$$
		\Phi_\pi\left(y\right) = \int_{\Theta} \sum_{e \in \mc E} \int_0^{(Ay'\pi(\theta))_e} \tau_e (\theta,s)ds d\Prob(\theta).
$$\smallskip
\end{proposition}
Furthermore, the equilibrium $f^{\pi,y^*(\pi)}$ is unique for every policy $\pi$.

\begin{proof}
Notice that $\Phi_\pi(y)$ can be written as
$$
\Phi_\pi(y) = \int_{\Theta} g(Ay'\pi(\theta),\theta) d\Prob(\theta),
$$ 
where $g(x,\theta) = \sum_e \int_0^{x_e} \tau_e(s,\theta) ds$ is convex in $x$ because the delay functions are non-decreasing.
Hence, $\Phi_\pi(y)$ is the combination of convex functions and is therefore convex. We now show that if $y^*(\pi)$ satisfies \eqref{eq:y}, then $f^{\pi,y^*(\pi)}(\theta) = A(y^*(\pi))'\pi(\theta)$ is a Bayesian user equilibrium. Since $y$ satisfies $y\1 = \1$ by construction, $y^*(\pi)$ satisfies \eqref{eq:y} if and only if, for every $i,j$ such that $y_{ij}^*(\pi)>0$, condition
$$
\frac{\partial \Phi_\pi(y^*(\pi))}{\partial y_{ik}} -  \frac{\partial \Phi_\pi(y^*(\pi))}{\partial y_{ij}} \ge 0
$$
holds true. Notice that
\begin{align*}
	\frac{\partial {\Phi_\pi(y)}}{\partial y_{ij}} &= \sum_e A_{ej} \int_{\Theta} \tau_e ((Ay'\pi(\theta))_e, \theta)\pi_i(\theta) d\Prob (\theta) \\
	&= \sum_e A_{ej} \int_{\Theta} \tau_e (f^{\pi,y}_e(\theta), \theta) d\Prob_i (\theta) \cdot \int_{\Theta} \pi_i(\theta) d\Prob (\theta) \\
	&= \E_i[c_j(f^{\pi,y}(\theta),\theta)] \cdot \int_{\Theta} \pi_i(\theta) d\Prob (\theta).
\end{align*}
Necessary and sufficient conditions for optimality state that if $y^*_{ij}(\pi)>0$ then for all $k$ in $\mc P$ 
$$
\E_i[c_k(f^{\pi,y^*(\pi)}(\theta),\theta)-c_j(f^{\pi,y^*(\pi)}(\theta),\theta)] \ge 0,
$$
which is the definition of Bayesian user equilibrium. The uniqueness of the equilibrium then follows from the monotonicity of the delay functions.
\end{proof}


\begin{remark}
	Observe that Proposition \ref{prp:pot} implies that $y^*(\pi)$ is the Nash equilibrium of a population game that admits a weighted potential. For more details, see \cite{sandholm2001potential}.
\end{remark}

For simplicity of notation, from now on we shall denote by $f^\pi$ the Bayesian user equilibrium corresponding to information policy $\pi$.
We express the cost of a policy $\pi$ as the expected total travel time at the corresponding Bayesian user equilibrium, i.e.,
$$
	C(\pi) = \int_\Theta \sum_ef_e^{\pi}(\theta)\tau_e\left(f_e^{\pi}(\theta),\theta\right) d\Prob(\theta).
$$
The system planner then wants to solve the problem
\begin{equation}\label{eq:problem1}
	\pi^* \in \argmin_{\pi: \Theta \to \Delta_\mc P} C(\pi).
\end{equation}
Problem \eqref{eq:problem1} can be reformulated by using the \emph{revelation principle} \cite{bergemann19,bergemann16a}, which states that given a policy $\pi$, there always exists a policy $\tilde\pi$ such that $C(\pi) = C(\tilde\pi)$ and $y^*(\tilde\pi) = \mathbf{I}$, hence $f^{\tilde\pi}(\theta)=A\tilde\pi(\theta)$. In other words, this means that restricting the attention to policies under which no user wants to deviate from the received recommendation is without loss of generality. This requirement is known as \emph{obedience constraint}. 
Revelation principle allows to formulate the information design problem as follows.
\begin{problem}\label{prob}
The optimal information policy is
	$$
		\pi^* = \argmin_{\pi: \Theta \to \Delta_\mc P} \int_{\Theta} \sum_{e \in \mc E} f_e^\pi(\theta) \tau_e(f_e^\pi(\theta),\theta)d\Prob(\theta)\\
	$$
	subject to
	\begin{gather*}
	\pi_i(\omega)\E_i[c_i (f^\pi(\theta),\theta) - c_j (f^{\pi}(\theta),\theta)] \le 0, \forall i,j \in \mc P, \forall \omega \in \Theta \\
	f^\pi(\theta) = A \pi(\theta).
	\end{gather*}
\end{problem}\smallskip

\begin{remark}
Multiplying the obedience constraints by $\pi_i(\omega)$ allows a distinction in two cases: if some user receives signal $i$, i.e., $\pi_i(\omega)>0$ for at least a $\omega$ in $\Theta$, then we impose that following the recommendation must be convenient; if no user receives signal $i$, i.e., $\pi_i(\omega)=0$ for every $\omega$, then the constraint is trivially satisfied. 
\end{remark}
\begin{remark}\label{remark:convex}
Observe that \eqref{eq:problem1} is a bi-level program, namely the planner optimizes the policy $\pi$, but the performance of the policy depends on $y^*(\pi)$, which in turn is the result of the convex program \eqref{eq:y} that depends on $\pi$. Instead, Problem~ \ref{prob} is a single-level optimization problem with obedience constraints, which is more tractable. Notice also that, if the delay functions are convex in the first argument, the objective function of Problem~\ref{prob} is convex in $\pi$ independently of the network topology and the prior distribution of the network state. In contrast, the obedience constraints are in general non-convex, and they convex if the network has two parallel links and affine delay functions.
\end{remark}

We now define the price of anarchy ($PoA$).
\begin{definition}[Price of anarchy]
	Let $\pi^*$ be the solution of Problem \ref{prob}. 
	Then,
	$$
		PoA = \frac{\int_{\Theta} \sum_{e \in \mc E} f_e^{\pi^*}(\theta) \tau_e(f_e^{\pi^*}(\theta),\theta)d\Prob(\theta)}{\int_{\Theta} \sum_{e \in \mc E} f_e^*(\theta) \tau_e(f_e^*(\theta),\theta)d\Prob(\theta)} \ge 1.
	$$
\end{definition}\medskip

In other words, in the context of information design the price of anarchy measures how suboptimal the optimal policy $\pi^*$ is compared to the full-information system optimum.

\section{Information design on two parallel links}\label{sec:results}
In this section we consider the information design problem in a network with two parallel links and affine delay functions. For this case, we first establish in Section \ref{sec:gen_prior} sufficient conditions on the support and on the moments of the random variables that describe the delay functions under which $PoA=1$, and apply these conditions to a case of interest in Section \ref{sec:a_det}. Then, we consider in Section \ref{sec:uniform} a case-study showing that optimality can still be achieved even if these sufficient conditions are not met.

\subsection{Arbitrary prior}\label{sec:gen_prior}
We consider a graph $\mc G=(\mc N,\mc E)$ with two nodes $\mc N = (o,d)$ and two parallel links from $o$ to $d$. We assume that the delay functions are affine in the flow, i.e.,
\begin{equation*}
	\tau_e(a_e,b_e,f_e)=a_ef_e+b_e, \quad e \in \{1,2\}, \\
\end{equation*}
where $a_e, b_e$ are non-negative random variables.
In this setting the policy is a map $\pi(a,b)$ such that $\pi_1(a,b)+\pi_2(a,b) = 1$, hence we can formulate the problem in terms of $\pi_1$. With this in mind, by noticing that the policy depends on $b$ via $x=b_1-b_2$, and rearranging the terms, the problem reads as follows.
\begin{problem}\label{prob:affine}
	Find $\pi_1^*:\Theta \to [0,1]$ to minimize
$$ \int_{\Theta}[(x-2a_2)\pi_1(a,x)+(a_1+a_2){\pi_1^2(a,x)}]d\mathbb{P}(a,x)
$$
under obedience constraints
\be\label{eq:ob1}	
		\int_{\Theta} [(x-a_2)\pi_1(a,x)+(a_1+a_2)\pi_1^2(a,x)]d \mathbb{P}(a,x) \le 0\,,\ee
		\be\ba{rcl}\ds\int_{\Theta} [a_2-x+(x-a_1-2a_2)\pi_1(a,x)\\\ds\qquad +(a_1+a_2)\pi_1^2(a,x) ]d\mathbb{P}(a,x) &\leq & 0\,.\ea
	\label{eq:ob2}\ee
\end{problem}\medskip


\begin{lemma}\label{lemma:tildepi}
Consider Problem \ref{prob:affine}. Then,
	$$
		f_1^*\left(a,x\right)= \left[\frac{2a_2-x}{2\left(a_1+a_2\right)}\right]_0^1
		\label{eq_uncpol}.
	$$
\end{lemma}\medskip

\begin{proof}
Notice that $f^*$ does not depend on the prior $d\Prob$, since $f^*(a,x)$ is found for every $a,x$ as the minimum of
$$
(x-2a_2)f_1(a,x)+(a_1+a_2){f_1^2(a,x)},
$$
which is convex for every $a,x$ since $a_e$ are non-negative. The claim follows immediately from computing the stationary point of the objective function with respect to $f_1$ and noticing that $f_1$ belongs to $[0,1]$ by definition.
\end{proof}

Observe that in networks with parallel links, under obedience constraints it holds $f^\pi = \pi$, because $A=\I$ and $z=\I$. Hence, if the full-information system optimum $f^*$ satisfies the obedience constraints, then $\pi^* = f^*$ and $PoA=1$, which means that private information design can achieve optimality. This idea is the core of the next result. 
To better formalize the next result, let us define the support of a random variable. Given a random variable $X$ in $\mathds{R}^n$, its support is the set of values that the random variable can take, i.e.,
$$
\supp(X) = \{x \in \mathds{R}^n: \mathbb{P}(B(x,\epsilon))>0 \ \forall \epsilon > 0\},
$$
where $B(x,\epsilon)$ is the open ball centered in $x$ with radius $\epsilon$.

\begin{theorem}\label{thm1}
	Consider Problem \ref{prob:affine}. Assume that:
	\begin{equation}\label{eq:supp1}
		\begin{cases}
			\max(\supp(x)) \le 2\min(\supp(a_2)) \\
			\min(\supp(x)) \ge -2\max(\supp(a_1)). 
		\end{cases}
	\end{equation}
	Then, $\pi^*= f^*$ and $PoA=1$ if and only if
	\begin{equation}\label{eq:momenti1}
		\begin{cases}
			\displaystyle \mathbb{E}\left[\frac{2a_2x-x^2}{\left(a_1+a_2\right)}\right] \leq 0\\[8pt]
			\displaystyle\mathbb{E}\left[\frac{-2a_1 x-x^2}{\left(a_1+a_2\right)}\right] \leq 0
		\end{cases}
	\end{equation}
\end{theorem}\medskip
\begin{proof} 
Condition \eqref{eq:supp1} guarantees that $f^*$ does not saturate, i.e., $f^*_1 = (2a_2-x)/2(a_1+a_2).$ Plugging $f^*_1$ into the obedience constraint \eqref{eq:ob1}, we get that the constraint is satisfied if and only if
\begin{equation*}
\int_{\Theta} \left((x-a_2)\frac{2a_2-x}{2(a_1+a_2)}+\frac{(2a_2-x)^2}{4(a_1+a_2)}\right)d\mathbb{P}(a,x) \le 0,
\end{equation*}
which leads by some computation to
$$
\mathbb{E}\left[\frac{2a_2x-x^2}{\left(a_1+a_2\right)}\right] \leq 0
$$
Likewise, we get for the second obedience constraint \eqref{eq:ob2},
$$
\mathbb{E}\left[\frac{-2a_1 x-x^2}{\left(a_1+a_2\right)}\right] \leq 0,
$$
concluding the proof.
\end{proof}

\begin{remark}
As already observed, Problem \ref{prob} is in general non-convex unless the network is composed of two parallel links with affine delay functions. Hence, finding the optimal information policy $\pi^*$ is not an easy task in general. However, the approach adopted in Theorem \ref{thm1} to verify whether the system-optimum flow can be induced by information design can be generalized for arbitrary network topologies. Indeed, if there exists a policy $\pi^*$ that satisfies the obedience constraints and is compatible with the system-optimum flow (i.e., $f^*(\theta) = A\pi^*(\theta)$), then any other policy cannot achieve a better cost by definition of $f^*$, proving that $\pi^*$ is the optimal policy and achieves $PoA=1$.
\end{remark}

Observe by Theorem \ref{thm1} that, under \eqref{eq:supp1}, if $\mathbb{E}\left[x\right] = 0$ (i.e., the two links have same expected free-flow delay), and $a$ and $x$ are independent, i.e., $\E[a_e x] = \E[a_e] \E[x]$, then the two obedience constraints in \eqref{eq:momenti1} are always satisfied and optimality is thus achieved by information policy $\pi^*_1 = f_1^* = (2a_2-x)/2(a_1+a_2)$. Instead, this does not hold true in general if $E[x]=0$, but $a$ and $x$ are not independent.
We remark that condition \eqref{eq:momenti1} is necessary and sufficient for optimality assuming that \eqref{eq:supp1} holds. Condition \eqref{eq:supp1} guarantees that the full-information system optimum does not saturate and allows to express conditions for optimality in terms of moments of the random variables. However, \eqref{eq:supp1} is not necessary for optimality. We shall analyze in Section \ref{sec:uniform} a case where \eqref{eq:supp1} is violated but optimality is achieved.

\subsection{Deterministic linear coefficients and arbitrary prior}\label{sec:a_det}
In this section we analyze Problem \ref{prob:affine} when $a$ is known and $b$ is the only random variable (hence, also $x=b_1-b_2$ is random). The next result follows from Theorem \ref{thm1}.
\begin{theorem}\label{thm2}
	Consider Problem \ref{prob:affine} with $a$ known. Assume that:
	\begin{equation}\label{eq:supp2}
		supp(x) \subseteq \left[-2a_1, 2a_2\right].
	\end{equation}
	Then, $\pi^*=f^*$ and $PoA=1$ if and only if
	\begin{equation}\label{eq:momenti2}
	 -\frac{\mathbb{E}[x^2]}{2a_1} \le \mathbb{E}[x] \le \frac{\mathbb{E}[x^2]}{2a_2}.
	\end{equation}
\end{theorem}
\par
\bigskip
\begin{proof}
The proof follows from Theorem \ref{thm1}. In particular, when $a$ is publicly known, \eqref{eq:supp1} is equivalent to \eqref{eq:supp2} and \eqref{eq:momenti1} is equivalent to \eqref{eq:momenti2}.
\end{proof}

\begin{remark}
It is interesting to establish a parallelism between the price of anarchy in the setting of Theorem \ref{thm2} and in the standard setting of routing games where $x$ and $a$ are known to the users. To this end, it proves useful to write \eqref{eq:momenti2} as
\be\label{eq:variance}
\begin{cases}
	\sigma_x^2 \ge E[x](2a_2-E[x]) \\
	\sigma_x^2 \ge -E[x](2a_1+E[x]).
\end{cases}
\ee
where $\sigma_x^2$ is the variance of $x$. If $\sigma_x^2=0$ (i.e., $x$ is known) and $x$ belongs to $[-2a_1,2a_2]$, then Theorem \ref{thm2} states that optimality can be achieved if and only if 
$$
x(2a_1+x) \ge 0 \ge x(2a_2-x)
$$
i.e., when either $x=2a_2$ or $x=-2a_1$ (all users select the same link at user-equilibrium flow and the equilibrium coincides with the system-optimum), or $x = 0$. For other intermediate values of $x$, if $\sigma_x^2=0$ optimality cannot be achieved. Instead, if $x$ is uncertain, optimality is achievable for intermediate values of $\E[x]$ also. In particular, \eqref{eq:variance} states that the larger the variance is, the easier is to achieve optimality through information design. This shows that the uncertainty on the network state is beneficial for the system, in the sense that an omniscient planner that discloses privately the information on the network state can leverage the uncertainty to persuade the users to distribute according to the system-optimum flow.
\end{remark}

\subsection{Deterministic linear coefficients and uniform prior}\label{sec:uniform}
The theoretical results in Theorems \ref{thm1} and \ref{thm2} assume that the support of the random variables is such that under the system-optimum flow $f^*$ the users always distribute on both the links. This assumption indeed allows us to state sufficient and necessary conditions for optimality in terms of moments of the random variables that characterize whether the obedience constraints are active or inactive. In this section we consider a case-study that allows explicit computation even relaxing the assumption on the support. In particular, we assume that
\begin{equation}
	d\Prob(b_1,b_2) =
	\begin{cases}
		1 \quad & \text{if} \ b_1,b_2 \in [0,1]^2 \\
		0 & \text{otherwise},
	\end{cases}
	\label{eq:unif}
\end{equation}
and assume that $a$ is given and non-negative. The next result characterizes the solution of the information design problem. 

\begin{theorem}\label{thm3}
	Consider Problem \ref{prob:affine}. Let the prior distribution of $b$ be as in \eqref{eq:unif}, and let $a$ be known.
	Then, $PoA=1$ for every $a$ in $\mathds{R}_+^2$.
\end{theorem}
\begin{proof}
	The proof is contained in the Appendix.
\end{proof}

Notice that specifying the prior distribution of $b$ determines the distribution of $x$. Moreover, notice that \eqref{eq:unif} implies that $\supp(x) = [-1,1]$. Hence, \eqref{eq:supp2} is satisfied if $a_1, a_2 \ge 1/2$. Nonetheless, Theorem \ref{thm2} states that optimality is achieved for every value of $a$, extending the results of Theorem \ref{thm2}.

\section{Conclusion}\label{sec:conclusion}
In this paper we study optimal information design in Bayesian routing games. The problem investigates how a planner that observes the stochastic state of a transportation network should privately disclose information to strategic users in order to influence their routing behavior, with the goal of steering the system towards the system optimum. Due to the revelation principle, the analysis focuses without loss of generality on private path recommendations under the assumption that users have no incentive in deviating from the recommendation they receive (obedience constraints). We formalize and analyze the general problem, and then focus on the case of two links and affine delay functions, establishing sufficient conditions on the moments and on the support of the random variables under which the system optimum can be achieved by information design. We then analyze a special case where optimality is achieved even if these sufficient conditions are not satisfied.

This work focuses on when optimality is achievable. Future research lines focus on quantifying how suboptimal is the best information policy when optimality cannot be achieved. Other research lines aim at expanding the analysis to more general network topologies and non-linear delay functions. Finally, it would be interesting to analyze the outcome of multiple information providers competing for customers.

\bibliographystyle{IEEEtran}
\bibliography{references.bib}

\begin{thebibliography}{10}
\providecommand{\url}[1]{#1}
\csname url@samestyle\endcsname
\providecommand{\newblock}{\relax}
\providecommand{\bibinfo}[2]{#2}
\providecommand{\BIBentrySTDinterwordspacing}{\spaceskip=0pt\relax}
\providecommand{\BIBentryALTinterwordstretchfactor}{4}
\providecommand{\BIBentryALTinterwordspacing}{\spaceskip=\fontdimen2\font plus
\BIBentryALTinterwordstretchfactor\fontdimen3\font minus
  \fontdimen4\font\relax}
\providecommand{\BIBforeignlanguage}[2]{{%
\expandafter\ifx\csname l@#1\endcsname\relax
\typeout{** WARNING: IEEEtran.bst: No hyphenation pattern has been}%
\typeout{** loaded for the language `#1'. Using the pattern for}%
\typeout{** the default language instead.}%
\else
\language=\csname l@#1\endcsname
\fi
#2}}
\providecommand{\BIBdecl}{\relax}
\BIBdecl

\bibitem{wardrop1952road}
J.~G. Wardrop, ``Road paper. some theoretical aspects of road traffic
  research.'' \emph{Proceedings of the institution of civil engineers}, vol.~1,
  no.~3, pp. 325--362, 1952.

\bibitem{cole2006much}
R.~Cole, Y.~Dodis, and T.~Roughgarden, ``How much can taxes help selfish
  routing?'' \emph{Journal of Computer and System Sciences}, vol.~72, no.~3,
  pp. 444--467, 2006.

\bibitem{maggistro2018stability}
R.~Maggistro and G.~Como, ``Stability and optimality of multi-scale
  transportation networks with distributed dynamic tolls,'' in \emph{2018 IEEE
  Conference on Decision and Control (CDC)}.\hskip 1em plus 0.5em minus
  0.4em\relax IEEE, 2018, pp. 211--216.

\bibitem{cianfanelli2023optimal}
L.~Cianfanelli, G.~Como, A.~Ozdaglar, and F.~Parise, ``Optimal intervention in
  transportation networks,'' \emph{IEEE Transactions on Automatic Control},
  2023.

\bibitem{farahani2013review}
R.~Z. Farahani, E.~Miandoabchi, W.~Y. Szeto, and H.~Rashidi, ``A review of
  urban transportation network design problems,'' \emph{European journal of
  operational research}, vol. 229, no.~2, pp. 281--302, 2013.

\bibitem{acemoglu2018informational}
D.~Acemoglu, A.~Makhdoumi, A.~Malekian, and A.~Ozdaglar, ``Informational
  braess' paradox: The effect of information on traffic congestion,''
  \emph{Operations Research}, vol.~66, no.~4, pp. 893--917, 2018.

\bibitem{wu2017informational}
M.~Wu, J.~Liu, and S.~Amin, ``Informational aspects in a class of bayesian
  congestion games,'' in \emph{2017 American Control Conference (ACC)}.\hskip
  1em plus 0.5em minus 0.4em\relax IEEE, 2017, pp. 3650--3657.

\bibitem{wu2021value}
M.~Wu, S.~Amin, and A.~E. Ozdaglar, ``Value of information in bayesian routing
  games,'' \emph{Operations Research}, vol.~69, no.~1, pp. 148--163, 2021.

\bibitem{daskam17}
S.~Das, E.~Kamenica, and R.~Mirka, ``Reducing congestion through information
  design,'' \emph{55th Annual Allerton Conference}, 2017.

\bibitem{tavafoghi2017informational}
H.~Tavafoghi and D.~Teneketzis, ``Informational incentives for congestion
  games,'' in \emph{2017 55th Annual Allerton Conference on Communication,
  Control, and Computing (Allerton)}.\hskip 1em plus 0.5em minus 0.4em\relax
  IEEE, 2017, pp. 1285--1292.

\bibitem{savla}
Y.~Zhu and K.~Savla, \emph{Information design in non-atomic routing games with
  partial participation: computation and properties}.\hskip 1em plus 0.5em
  minus 0.4em\relax University of Southern California, 2021.

\bibitem{dughmi2016algorithmic}
S.~Dughmi and H.~Xu, ``Algorithmic bayesian persuasion,'' in \emph{Proceedings
  of the forty-eighth annual ACM symposium on Theory of Computing}, 2016, pp.
  412--425.

\bibitem{meigs2020optimal}
E.~Meigs, F.~Parise, A.~Ozdaglar, and D.~Acemoglu, ``Optimal dynamic
  information provision in traffic routing,'' \emph{arXiv preprint
  arXiv:2001.03232}, 2020.

\bibitem{tavafoghi2019strategic}
H.~Tavafoghi and D.~Teneketzis, ``Strategic information provision in routing
  games,'' 2019.

\bibitem{bergemann16a}
D.~Bergemann and S.~Morris, ``Bayes correlated equilibrium and the comparison
  of information structures in games,'' 2016.

\bibitem{bergemann19}
------, ``Information design: A unified perspective,'' \emph{Journal of
  Economic Literature}, vol.~57, no.~1, pp. 44--95, 2019.

\bibitem{sandholm2001potential}
W.~H. Sandholm, ``Potential games with continuous player sets,'' \emph{Journal
  of Economic theory}, vol.~97, no.~1, pp. 81--108, 2001.

\end{thebibliography}

\appendix
\subsection* {Proof of Theorem \ref{thm3}}\label{app}
As proved in Lemma \ref{lemma:tildepi}, the solution of the relaxed problem without obedience constraints is
$$
\tilde{\pi}_1\left(a,b\right)= \left[\frac{2a_2-b_1+b_2}{2\left(a_1+a_2\right)}\right]_0^1.
$$
To prove the statement, we need to verify that $\tilde{\pi}_1$ satisfies the obedience constraints \eqref{eq:ob1}, \eqref{eq:ob2} and is therefore the solution of the information design problem. The two obedience constraint are in the form
\begin{equation}\label{eq:integral}
	\begin{aligned}
&\int_{[0,1]} \int_{[0,1]} g(a,b) db_1 db_2 \le 0,\\
&\int_{[0,1]} \int_{[0,1]} h(a,b) db_1 db_2 \le 0
\end{aligned}
\end{equation}
where $g$ refers to obedience constraint \eqref{eq:ob1} and $h$ to the second one \eqref{eq:ob2}. We focus on $g$, but similar considerations hold for $h$. In particular,
\begin{equation}\label{eq:g}
g(a,b) = (b_1-b_2-a_2)\tilde\pi_1(a,b)+(a_1+a_2)\tilde\pi_1^2(a,b).
\end{equation} 
Observe that $\tilde\pi$ and $g$ depend on $b$ via the difference $x=b_1-b_2$ and $a$ is a constant (similar conditions hold for the second obedience constraint, see \eqref{eq:ob1}, \eqref{eq:ob2}).
For these reasons, from now on, we shall omit the dependence of $g$ on $a$, and write $g(x)$ instead of $g(a,b)$. To simplify the notation, let us define $$\alpha = \frac{1}{2\left(a_1+a_2\right)}, \quad \beta = \frac{a_2}{a_1+a_2}.$$
\begin{figure}
	\centering
	\includegraphics[width=0.98\linewidth]{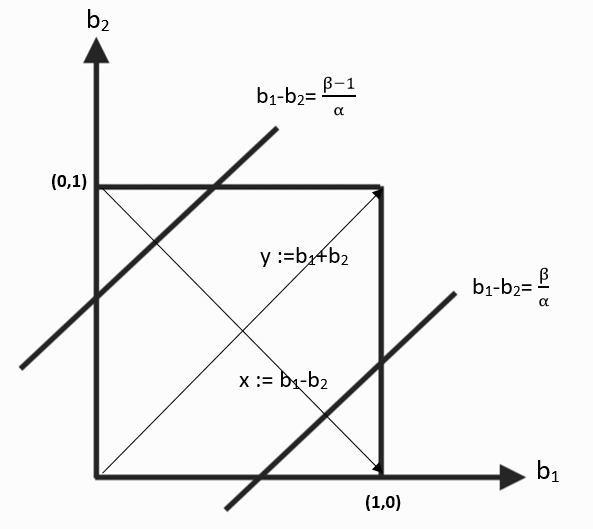}
	\caption{Representation of the space of the variables $b_1$ and $b_2$.}
	\label{fig:simplesso}
\end{figure}
With this change of variables, 
\begin{equation}
	\tilde\pi_1\left(x\right)=
	\begin{cases}
		0 & x \geq \frac{\beta}{\alpha} \\
		\beta-\alpha x & \frac{\beta-1}{\alpha} < x < \frac{\beta}{\alpha} \\
		1 & x \leq \frac{\beta-1}{\alpha}  \\
	\end{cases}
	\label{eq:pi_3}
\end{equation} 
Since $g$ depends on $\tilde\pi$, the structure of $\tilde\pi_1$ suggests a split of the integral over $[0,1]^2$ in three areas, as shown in Figure \ref{fig:simplesso}.
We now define $y=b_1+b_2$, and make a change of variables from $(b_1,b_2)$ to $(x,y)$. Notice that $g$ does not depend on $y$. The first integral in \eqref{eq:integral} thus can be written as
\begin{equation}\label{eq:integral3}
\frac{1}{2}\Big(\int_{T_1} g_{T_1}(x) dx dy + \int_R g_R(x) dx dy + \int_{T_2}g_{T_2}(x) dx dy\Big),
\end{equation}
where the factor $1/2$ comes from the change of variables, $T_1$ is the top-left triangle with $x \le (\beta-1)/\alpha$, $T_2$ is the bottom-right triangle with $x \ge \beta/\alpha$, $R$ is the area in the middle, and $g$ in each area is found by plugging \eqref{eq:pi_3} into \eqref{eq:g}. We also assume $a_1 \le a_2$ without loss of generality. Observe that depending on the values of $\beta, \alpha$ (or equivalently $a_1, a_2$) the triangles $T_1$ and $T_2$ might not exist. Three different cases need to be studied:
\begin{enumerate}
	\item There are both $T_1$ and $T_2$, i.e.,
	\begin{equation}\label{eq:case1}
	\frac{\beta-1}{\alpha}=-2a_1 > -1, \quad \frac{\beta}{\alpha}=2a_2 < 1;
	\end{equation}
	\item There is $T_1$ but not $T_2$, i.e., 
	$$
	\frac{\beta-1}{\alpha}=-2a_1 > -1, \quad \frac{\beta}{\alpha}=2a_2 \ge 1;
	$$
	\item There are not $T_1$ and $T_2$, i.e., 
	$$
	\frac{\beta-1}{\alpha}=-2a_1 \le -1, \quad \frac{\beta}{\alpha}=2a_2 \ge 1.
	$$
	
\end{enumerate}
Observe that the case with $T_2$ and without $T_1$ is excluded by the assumption $a_1 \le a_2$, which is without loss of generality. In case 1), the integrals on the three areas can be written as follows.
\begin{equation}\label{eq:T1}
\frac{1}{2}\int_{T_1} g_{T_1}(x) dx dy = \frac{1}{2}\int_{-1}^{\frac{\beta-1}{\alpha}} \int_{-x}^{x+2} g_{T_1}(x)dydx,
\end{equation}
with $g_{T_1}(x) = a_1+x$,
\begin{equation}\label{eq:T2}
	\frac{1}{2}\int_{T_2} g_{T_2}(x) dx dy = \frac{1}{2}\int_{\frac{\beta}{\alpha}}^{1}\int_{x}^{x+2}g_{T_2}(x)dydx,
\end{equation}
with $g_{T_2}(x) = 0$, and
\begin{equation}\label{eq:R}
	\begin{aligned}
	\frac{1}{2}\int_{R} g_{R}(x) dxdy &=\int_{0}^{\frac{1-\beta}{\alpha}}\int_{-y}^{+y}g_R(x)dxdy \\ &+\int_{\frac{1-\beta}{\alpha}}^{\frac{\beta}{\alpha}}\int_{\frac{\beta-1}{\alpha}}^{+y}g_R(x)dxdy \\ &+\int_{\frac{\beta}{\alpha}}^{1}\int_{\frac{\beta-1}{\alpha}}^{\frac{\beta}{\alpha}}g_R(x)dxdy
	\end{aligned}
\end{equation}
with $g_R(x)=\frac{\beta^2}{2\alpha} -a_2\beta + a_2 \alpha x -\frac{1}{2}\alpha x^2$.
\begin{figure}
	\centering
	\includegraphics[width=0.8\linewidth]{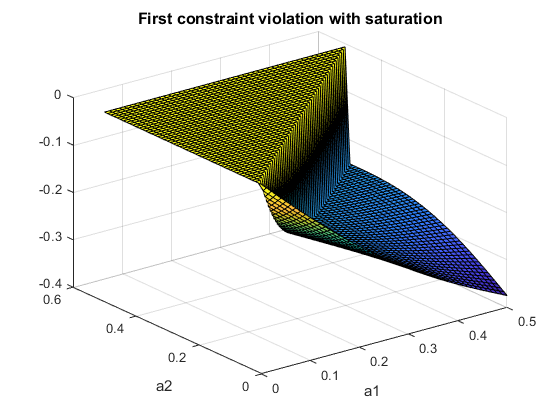}
	\caption{Plot of \eqref{eq:cons1sat} with $a$ in $\left[0,\frac{1}{2} \right]^2$.}
	\label{fig:fir_const}
\end{figure} 
\begin{figure}
	\centering
	\includegraphics[width=0.8\linewidth]{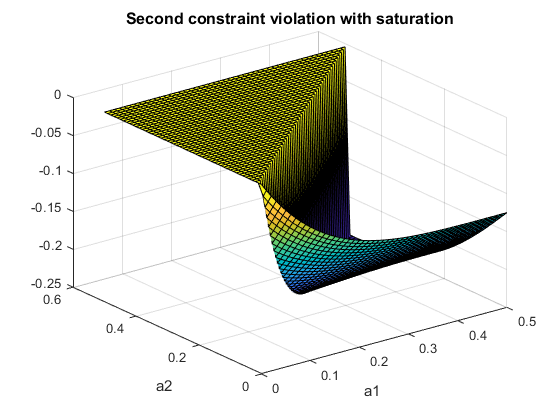}
	\caption{Plot of \eqref{eq:cons2sat} with $a$ in $\left[0,\frac{1}{2}\right]^2$.}
	\label{fig:sec_const}
\end{figure} 
Plugging \eqref{eq:T1}, \eqref{eq:T2} and \eqref{eq:R} into \eqref{eq:integral3} it turns out that the first constraint is violated if and only if
\begin{equation*}
	\begin{aligned} 
		&\frac{1}{2}\int_{-1}^{\frac{\beta-1}{\alpha}} \int_{-x}^{x+2}(a_1+x) dydx  +\\
		&\int_{0}^{\frac{1-\beta}{\alpha}}\int_{-y}^{+y} \left(\frac{\beta^2}{2\alpha} -a_2\beta + a_2 \alpha x -\frac{1}{2}\alpha x^2\right) dxdy +\\
		&\int_{\frac{1-\beta}{\alpha}}^{\frac{\beta}{\alpha}}\int_{\frac{\beta-1}{\alpha}}^{+y} \left(\frac{\beta^2}{2\alpha} -a_2\beta + a_2 \alpha x -\frac{1}{2}\alpha x^2\right) dxdy +\\
		&\int_{\frac{\beta}{\alpha}}^{1}\int_{\frac{\beta-1}{\alpha}}^{\frac{\beta}{\alpha}} \left(\frac{\beta^2}{2\alpha} -a_2\beta + a_2 \alpha x -\frac{1}{2}\alpha x^2\right) dxdy +\\ 
		&\frac{1}{2}\int_{\frac{\beta}{\alpha}}^{1}\int_{x}^{-x+2} 0 dydx >0,
	\end{aligned}
\end{equation*}
By some computation, this proves equivalent to
\begin{equation}
	\begin{aligned}
	&2a_1^4-2a_2^4-4a_1^3+2a_2^3+4a_1^3a_2+\\
	+&3a_1^2-6a_1^2a_2-a_1-a_2+3a_1a_2 >0.
	\end{aligned}
	\label{eq:cons1sat}
\end{equation}
Observe that \eqref{eq:cons1sat} must be evaluated under \eqref{eq:case1}, i.e., $a_1,a_2<1/2.$
It is possible to show that \eqref{eq:cons1sat} is never satisfied in the domain, i.e. the constraint is satisfied, as illustrated in Figure \ref{fig:fir_const}. By analogous computation, by replacing $h$ with $g$, we get that the second obedience constraint \eqref{eq:ob2} is violated if and only if
\begin{equation}	\label{eq:cons2sat}
	\begin{aligned}
	&-2a_1^4+2a_2^4+2a_1^3-4a_2^3+4a_1a_2^3+\\
	&+3a_2^2-6a_1a_2^2-a_1-a_2+3a_1a_2 >0.
	\end{aligned}
\end{equation}
It is possible to show that \eqref{eq:cons2sat}
is never satisfied in the domain, i.e. the constraint is satisfied, as illustrated in Figure \ref{fig:sec_const}. This shows that in case 1) both the constraints are satisfied by $\tilde\pi_1$, which in turn implies that $PoA=1$.

We now consider case 2), in which there exists $T_1$ but not $T_2$. The condition for the first constraint to be violated is
\begin{equation*}
	\begin{aligned} 
		&\frac{1}{2}\int_{-1}^{\frac{\beta-1}{\alpha}} \int_{-x}^{x+2}(a_1+x) dydx  +\\
		&\int_{0}^{\frac{1-\beta}{\alpha}}\int_{-y}^{+y} \left(-a_2\beta +\frac{\beta^2}{2\alpha} + a_2 \alpha x -\frac{1}{2}\alpha x^2\right) dxdy +\\
		&\int_{\frac{1-\beta}{\alpha}}^{1}\int_{\frac{\beta-1}{\alpha}}^{+y} \left(-a_2\beta +\frac{\beta^2}{2\alpha} + a_2 \alpha x -\frac{1}{2}\alpha x^2\right) dxdy >0,
	\end{aligned} 
\end{equation*}
where the first term refers to the integral on $T_1$, and the last two to the integral on $R$. Rearranging the terms, we get that the first constraint is satisfied if and only if
$$
	\begin{aligned}
	&4a_1^4-8a_1^3+8a_1^3a_2+6a_1^2-12a_1^2a_2\\
	-&2a_1-a_2+6a_1a_2 -\frac{1}{4} > 0,
	\end{aligned}
$$
or, equivalently
\begin{equation}\label{eqeq}
	a_2 < \frac{-4a_1^4+8a_1^3-6a_1^2+2a_1+\frac{1}{4}}{\left(2a_1-1\right)^3}.
\end{equation}
However, since we are in case $a_1<1/2$, the right-hand side of this equation is always negative. On the other hand, $a_2 \ge 1/2$, which implies that \eqref{eqeq} cannot be satisfied and the first obedience constraint is always satisfied by $\tilde\pi$ in case 2). Similar consideration hold for the second constraint.

Optimality in case 3) follows from Theorem \ref{thm2}. Indeed, in this case optimality is guaranteed by the fact that $a_1, a_2 \ge 1/2$, which implies \eqref{eq:supp2}. This concludes the proof.

\end{document}